\documentclass[conference]{IEEEtran}
\IEEEoverridecommandlockouts
\usepackage{cite}
\usepackage{amsmath,amssymb,amsfonts,amsthm}
\newtheorem{theorem}{Theorem}
\newtheorem{corollary}{Corollary}
\usepackage{algorithmic}
\usepackage{graphicx}
\usepackage{booktabs}
\usepackage{multirow}
\usepackage{textcomp}
\usepackage{tikz}
\usetikzlibrary{shapes.geometric, arrows, positioning}

\tikzstyle{component} = [rectangle, rounded corners, minimum width=2.5cm, minimum height=1cm, text centered, draw=black, fill=purple!30]
\tikzstyle{bus} = [rectangle, rounded corners, minimum width=4cm, minimum height=1cm, text centered, draw=black, fill=blue!20]
\tikzstyle{arrow} = [thick,->,>=stealth]
\tikzstyle{data} = [draw, fill=green!20, text centered, rounded corners, minimum height=0.8cm, minimum width=2cm]

\usepackage[colorlinks=true,linkcolor=black,citecolor=black,urlcolor=blue]{hyperref}
\usepackage[capitalize,nameinlink,noabbrev]{cleveref}
\crefname{figure}{Fig.}{Figs.}  
\Crefname{figure}{Fig.}{Figs.}  
\usepackage{xcolor}
\usepackage{listings}
\usepackage{url}
\usepackage{pifont} 

\begin{document}

\title{\textit{push0}: Scalable and Fault-Tolerant Orchestration\\for Zero-Knowledge Proof Generation}
\author{
    \IEEEauthorblockN{Mohsen Ahmadvand, Rok Pajni\v{c}, Ching-Lun Chiu}
    \IEEEauthorblockA{\textit{Zircuit} \\
    \{mohsen, rok, juno\}@zircuit.com}
}

\maketitle

\begin{abstract}
Zero-knowledge proof generation imposes stringent timing and reliability constraints on blockchain systems.
For ZK-rollups, delayed proofs cause finality lag and economic loss; for Ethereum's emerging L1 zkEVM, proofs must complete within the 12-second slot window to enable stateless validation.
The Ethereum Foundation's Ethproofs initiative coordinates multiple independent zkVMs across proving clusters to achieve real-time block proving, yet no principled orchestration framework addresses the joint challenges of (i) strict head-of-chain ordering, (ii) sub-slot latency bounds, (iii) fault-tolerant task reassignment, and (iv) prover-agnostic workflow composition.

We present \textit{push0}, a cloud-native proof orchestration system that decouples prover binaries from scheduling infrastructure.
push0 employs an event-driven dispatcher--collector architecture over persistent priority queues, enforcing block-sequential proving while exploiting intra-block parallelism.
We formalize requirements drawn from production ZK-rollup operations and the Ethereum real-time proving specification, then demonstrate via production Kubernetes cluster experiments that push0 achieves 5\,ms median orchestration overhead with 99--100\% scaling efficiency at 32 dispatchers for realistic workloads---overhead negligible ($<$0.1\%) relative to typical proof computation times of 7+ seconds.
Controlled Docker experiments validate these results, showing comparable performance (3--10\,ms P50) when network variance is eliminated.
Production deployment on the Zircuit zkrollup (14+ million mainnet blocks since March 2025) provides ecological validity for these controlled experiments.
Our design enables seamless integration of heterogeneous zkVMs, supports automatic task recovery via message persistence, and provides the scheduling primitives necessary for both centralized rollup operators and decentralized multi-prover networks.
\end{abstract}


\section{Introduction}

Zero-knowledge proofs (ZKPs) have emerged as a foundational primitive for blockchain scalability and privacy~\cite{sun2024zkp_survey}.
In the context of Ethereum, ZKPs serve two distinct but related purposes: enabling \emph{ZK-rollups} that batch transactions off-chain and submit succinct validity proofs to the L1~\cite{gorzny2024rollup}, and powering the forthcoming \emph{L1 zkEVM} that allows validators to verify block execution via proof checking rather than re-execution~\cite{ef_realtime_proving2025}.
Both applications impose hard timing constraints on proof generation---rollups require proofs for timely finality, while the L1 zkEVM targets proof generation within Ethereum's 12-second slot window.

\paragraph{ZK-Rollups and Finality}
Rollups aggregate transactions off the Ethereum mainnet, reducing gas costs while inheriting L1 security guarantees~\cite{l2beat_scaling_summary}.
ZK-rollups generate cryptographic proofs attesting to correct state transitions; these proofs, verified on-chain, enable trustless withdrawals without multi-day challenge periods~\cite{kalodner2018arbitrum,gorzny2022ideal}.
However, proof generation is computationally intensive, requiring specialized hardware (GPUs, FPGAs, or ASICs) and careful pipeline orchestration~\cite{ma2023gzkp,batchzk2025}.
If the proving pipeline stalls, the rollup cannot advance---a condition termed \emph{chain halt}---freezing user funds and blocking new transactions~\cite{gorzny2022ideal}.
Even without complete failure, proving delays directly translate to \emph{finality lag}, degrading user experience and increasing economic risk.

\paragraph{Ethereum L1 Real-Time Proving}
The Ethereum Foundation has articulated a roadmap for integrating zkEVM verification directly into the consensus layer~\cite{ef_realtime_proving2025}.
Under this model, validators would verify execution payloads by checking ZK proofs generated by one or more zkVMs, rather than re-executing all transactions.
The design target defines \emph{real-time proving} as completing proof generation for 99\% of mainnet blocks within 10 seconds (P99), leaving margin for network propagation within the 12-second slot~\cite{ef_realtime_proving2025}.
To ensure defense-in-depth akin to existing client diversity, the architecture envisions validators verifying proofs from multiple independent zkVM implementations.

\paragraph{The Ethproofs Coordination Challenge}
The Ethproofs platform, backed by the Ethereum Foundation, operationalizes this multi-prover vision~\cite{ethproofs2025}.
Ethproofs coordinates multiple zkVM implementations and proving clusters, demonstrating proving times competitive with Ethereum's 12-second slot time.
Recent advances---notably Succinct's SP1 Hypercube achieving 93\% real-time coverage with 200 GPUs~\cite{sp1_hypercube2025}---validate the feasibility of sub-slot proving.
However, Ethproofs exposes critical orchestration gaps: (i) no standardized mechanism ensures strict block ordering across heterogeneous provers; (ii) failover and task reassignment remain ad-hoc; (iii) integrating new zkVMs requires bespoke infrastructure; and (iv) proving workflows vary substantially across proof systems (STARK, SNARK, recursive composition), complicating unified scheduling.

\paragraph{Problem Statement}
Both ZK-rollup operators and L1 real-time proving networks require proof orchestration infrastructure that satisfies six properties:
\begin{enumerate}
    \item \textbf{Head-of-chain ordering}: Proofs must be generated in block-sequential order to maintain chain consistency.
    \item \textbf{Bounded latency}: Orchestration overhead must be negligible relative to proof computation to meet slot deadlines.
    \item \textbf{Fault tolerance}: Failed or slow proving tasks must be reassigned without manual intervention.
    \item \textbf{Prover agnosticism}: The system must accommodate arbitrary proof systems and workflow DAGs without code changes.
    \item \textbf{Horizontal scalability}: Components must scale elastically with proving load while maintaining correctness guarantees.
    \item \textbf{Observability}: Operators require visibility into queue depths, task states, and latency distributions for debugging and capacity planning.
\end{enumerate}
Existing solutions address subsets of these requirements, but \emph{no prior system satisfies all six simultaneously}.
Distributed proving systems (DIZK~\cite{zhou2018dizk}, SHARP~\cite{sharp2023}) scale computation but assume centralized, homogeneous infrastructure---lacking prover agnosticism and horizontal scalability across heterogeneous clusters.
Decentralized prover markets (Snarktor~\cite{garoffolo2024snarktor}, zkCloud~\cite{zkcloud_docs}, Succinct Network~\cite{succinct2025}) distribute capacity but lack strict ordering guarantees and operator-controlled failover; their economic mechanisms prioritize throughput over latency bounds.
GPU cluster schedulers (Gandiva~\cite{xia2018gandiva}, Tiresias~\cite{gu2019tiresias}) optimize throughput for ML workloads but do not model ZK-specific constraints such as recursive proof composition, block-level dependencies, or cryptographic barrier synchronization.
The gap exists because ZK proving occupies an unusual design point: it requires both the strict ordering of transactional systems and the parallel scalability of batch analytics, while accommodating heterogeneous proof systems with fundamentally different computational DAGs.

\paragraph{Contributions}
We present \textit{push0}, the first proof orchestration framework that jointly satisfies all six requirements for production ZK-rollups and L1 real-time proving.
Our contributions are:
\begin{itemize}
    \item \textbf{Requirements formalization}: We systematize six orchestration requirements distilled from nearly one year of operating a production ZK-rollup and analysis of the Ethereum real-time proving specification (\S\ref{sec:requirements}).
    These requirements expose a fundamental tension: head-of-chain ordering constrains parallelism, yet real-time proving demands maximal concurrency.
    \item \textbf{Dispatcher--collector architecture}: We resolve this tension through an event-driven design that decouples \emph{intra-block parallelism} (unlimited) from \emph{inter-block ordering} (strict).
    Persistent priority queues enforce block-sequential proof submission while enabling arbitrary concurrency within each block's proving DAG (\S\ref{sec:design}).
    \item \textbf{Correctness properties}: We establish starvation freedom, barrier locality, and routing continuity---properties that ensure correct behavior under horizontal scaling and component failures.
    We formalize the partition affinity requirement for ZK proof aggregation, quantifying the catastrophic impact of naive load balancing on barrier completion rates (\S\ref{sec:design}).
    \item \textbf{Prover-agnostic orchestration}: push0's orchestration layer is fully prover-agnostic, requiring adaptation only at the prover invocation boundary (minimal gluing code for CLI provers, Rust trait implementation for library-based provers).
    We validate this by deploying the same orchestration software across three distinct proof systems---each with different proving pipelines, hardware requirements, and orchestration topologies---demonstrating that prover-specific complexity is isolated to integration boundaries (\S\ref{sec:evaluation}).
    \item \textbf{Production deployment}: push0 has been deployed on the Zircuit zkrollup since March 4, 2025, processing over 14 million mainnet blocks. Production evaluation with SP1 GPU provers demonstrates ecological validity for the controlled experiments (\S\ref{sec:evaluation}).
\end{itemize}

This work addresses infrastructure often overlooked in ZK systems research.
While prior work focuses on proof system cryptography or circuit optimization, reliable proof orchestration is equally critical: a zkEVM producing proofs in 10 seconds provides no value if the orchestration layer cannot schedule, monitor, and recover proving tasks at scale.
Our experience shows that orchestration failures---not prover bugs---account for the majority of production incidents in ZK infrastructure.

\section{Related Work}

Prior work on distributed proving falls into three categories: general-purpose schedulers, proof-system-specific solutions, and decentralized prover markets.
None provides a general-purpose orchestration layer for heterogeneous provers with formal correctness guarantees.

\paragraph{General-Purpose Schedulers}
MapReduce~\cite{dean2008mapreduce}, Borg~\cite{verma2015large}, and Kubernetes provide fault-tolerant task scheduling for data processing and containerized workloads.
GPU schedulers like Gandiva~\cite{xia2018gandiva}, Tiresias~\cite{gu2019tiresias}, and Pollux~\cite{qiao2021pollux} optimize utilization for ML training.
These systems assume tasks are independent or follow ML-specific patterns (iterative gradient descent, data parallelism).
ZK state-transition proving requires different semantics: strict sequential ordering across blocks, barrier synchronization for recursive composition, and exactly-once execution (duplicate proofs waste expensive compute).
Deploying ZK provers on Kubernetes requires custom operators that re-implement scheduling logic; the platform provides containers, not orchestration semantics.

\paragraph{Proof-System-Specific Solutions}
DIZK~\cite{zhou2018dizk} distributes zkSNARK generation across clusters but assumes homogeneous workers running identical Groth16 code.
BatchZK~\cite{batchzk2025} and GZKP~\cite{ma2023gzkp} accelerate MSM and FFT kernels on GPUs---optimizing \emph{within} a prover, not \emph{across} provers.
StarkWare's SHARP~\cite{sharp2023} aggregates STARK proofs recursively but is tightly coupled to StarkWare's proof system and infrastructure.
These solutions scale throughput for a single proof system but cannot orchestrate heterogeneous provers (different zkVMs, different proof systems, different hardware) under a unified interface.

\paragraph{Decentralized Prover Markets}
Snarktor~\cite{garoffolo2024snarktor} coordinates recursive proof merges with economic incentives.
zkCloud~\cite{zkcloud_docs} and Succinct Network~\cite{succinct2025} propose marketplaces for proof generation.
CrowdProve~\cite{crowdprove2025} enables crowd-sourced subproof generation; Proo$\varphi$~\cite{proophi2024} formalizes incentive-compatible auctions.
These systems address \emph{capacity allocation} (who proves, at what price) but not \emph{orchestration} (task routing, barrier synchronization, failover).
Market mechanisms optimize for throughput and cost; they do not guarantee bounded latency or strict ordering---critical for real-time proving where P99 latency determines slot compliance.

\paragraph{Positioning}
push0 operates at a different abstraction level than prior work.
It does not compete with proof acceleration (DIZK, BatchZK), capacity markets (Snarktor, Succinct), or container orchestration (Kubernetes).
Instead, it provides the \emph{scheduling semantics} missing from all three:
\begin{itemize}
    \item \emph{Any prover}: Treats provers as opaque executables; integrates zkVMs (SP1, RISC Zero), circuit-specific provers (Halo2, Plonky2), and proprietary systems without code changes.
    \item \emph{Any infrastructure}: Runs on bare metal, Kubernetes, cloud VMs, or hybrid deployments; the message bus abstraction decouples orchestration from deployment.
    \item \emph{Correctness properties}: Establishes starvation freedom (\S\ref{sec:design}), barrier locality (Theorem~\ref{thm:locality}), and routing continuity---properties absent from prior systems.
\end{itemize}
push0 can serve as the orchestration layer beneath decentralized markets or above raw proving infrastructure, providing the scheduling guarantees that neither layer offers.

\section{System Design}
\label{sec:design}

We derive orchestration requirements from two sources: operational experience running a production ZK-rollup and the Ethereum Foundation's L1 real-time proving specification~\cite{ef_realtime_proving2025}.
The resulting design applies to both centralized rollup operators and decentralized multi-prover networks such as Ethproofs~\cite{ethproofs2025}.

\subsection{Requirements}
\label{sec:requirements}

\subsubsection{R1: Head-of-Chain Ordering}
State proofs must be generated in strict block-sequential order.
For ZK-rollups, the L1 verifier contract accepts proofs only for the next expected block; out-of-order submissions revert.
For L1 real-time proving, validators expect proofs corresponding to the current slot's execution payload.
The orchestration layer must enforce this ordering constraint regardless of task completion order.

\subsubsection{R2: Bounded Latency}
The Ethereum real-time proving design targets P99 latency under 10 seconds~\cite{ef_realtime_proving2025}.
Real-time proving targets require orchestration overhead to remain negligible compared to proof computation times (10+ seconds)~\cite{ethproofs2025}.
We target sub-100\,ms for per-task orchestration overhead including dispatch, routing, and collection.
For rollups, while slot deadlines are less stringent, finality lag directly impacts user experience and capital efficiency.

\subsubsection{R3: Fault Tolerance}
Proving infrastructure exhibits diverse failure modes: prover crashes, GPU memory exhaustion, network partitions, and timeouts on complex blocks.
A missed proof causes chain halt for rollups or missed slot for L1 validators.
The system must detect stalled tasks, reassign them automatically, and prioritize older blocks to prevent cascading delays.

\subsubsection{R4: Prover Agnosticism}
Ethproofs coordinates multiple zkVM implementations~\cite{ethproofs2025}, each with different internal workflows (witness generation, polynomial commitment, recursive aggregation).
The orchestration layer must treat provers as opaque executables parameterized only by input/output interfaces, enabling heterogeneous zkVM integration without infrastructure changes.

\subsubsection{R5: Horizontal Scalability}
Block complexity varies by orders of magnitude---from minimal transfers to gas-limit-saturating computation.
The system must elastically scale proving capacity while maintaining ordering guarantees.
Dispatchers must be stateless for arbitrary preemption; collectors may maintain soft state but must recover from message bus replay on restart.

\subsubsection{R6: Observability}
Debugging requires visibility into task states, queue depths, and latency distributions.
The system must expose metrics via standard interfaces (Prometheus, OpenTelemetry) and alert on anomalous conditions such as queue buildup or repeated task failures.

\subsection{Architecture Overview}

push0 employs an event-driven architecture with three core components: a persistent \emph{message bus}, stateless \emph{dispatchers}, and \emph{collectors} with reconstructible soft state.
This design fuses MapReduce-style parallel execution~\cite{dean2008mapreduce} with priority queuing for ordering enforcement.

A proving \emph{workflow} specifies a directed acyclic graph (DAG) of tasks.
Each node represents a computation (e.g., witness generation, sub-proof generation, aggregation proof composition); edges encode dependencies.
Independent nodes execute in parallel; dependent nodes await predecessor completion.
The system accepts two inputs: (1) a prover binary conforming to a standard I/O interface, and (2) a workflow specification.
This separation enables operators to upgrade provers or modify workflows without changing orchestration code.

\Cref{fig:push-zero-design} illustrates the architecture.
When a new block requires proving, the system enqueues a proof request.
Dispatchers consume tasks, invoke prover binaries, and publish results.
Collectors aggregate partial results using pluggable \emph{collection strategies} and enqueue successor tasks to the appropriate output queue.
The cycle continues until the final proof is produced.
The workflow DAG is implicit in the queue topology and strategy configuration---there is no explicit DAG data structure.

\begin{figure}[htbp]
    \centering
    \begin{tikzpicture}[
        node distance=0.8cm,
        box/.style={rectangle, rounded corners, minimum width=2.2cm, minimum height=0.7cm, text centered, draw=black, font=\scriptsize},
        innerbox/.style={rectangle, rounded corners, minimum width=1.6cm, minimum height=0.5cm, text centered, draw=gray!60, fill=gray!15, font=\tiny},
        arr/.style={thick,->,>=stealth}
    ]
        \node (messagebus) [box, fill=blue!15] {Message Bus};

        \node (dispatcher) [box, fill=purple!15, above=0.9cm of messagebus] {Dispatcher};
        \node (prover) [box, fill=green!15, right=0.6cm of dispatcher] {Prover};

        \node (collector) [rectangle, rounded corners, minimum width=3.6cm, minimum height=1.1cm, draw=black, fill=purple!15, below=0.9cm of messagebus] {};
        \node[font=\scriptsize, anchor=north west] at ([xshift=0.15cm, yshift=-0.05cm]collector.north west) {Collector};
        \node (strategy) [innerbox] at ([xshift=0.4cm, yshift=-0.05cm]collector.center) {Strategy};

        \node (block) [box, fill=yellow!20, left=0.6cm of messagebus] {Block $n$};

        \draw [arr] (block) -- (messagebus) node[midway, above, font=\tiny] {request};
        \draw [arr] (messagebus) -- (dispatcher) node[midway, right, font=\tiny] {task};
        \draw [arr] (dispatcher) -- (prover) node[midway, above, font=\tiny] {invoke};
        \draw [arr] (prover.south) -- ++(0,-0.3) -| (messagebus.north east) node[pos=0.25, right, font=\tiny] {result};
        \draw [arr] (messagebus) -- (collector) node[midway, right, font=\tiny] {aggregate};
        \draw [arr] (collector.north east) -- ++(0.3,0.3) -| (messagebus.south east) node[pos=0.25, right, font=\tiny] {emit};
    \end{tikzpicture}
    \caption{push0 architecture: message bus coordinates dispatchers (which invoke provers) and collectors (which aggregate via pluggable strategies). The strategy is internal to the collector.}
    \label{fig:push-zero-design}
\end{figure}
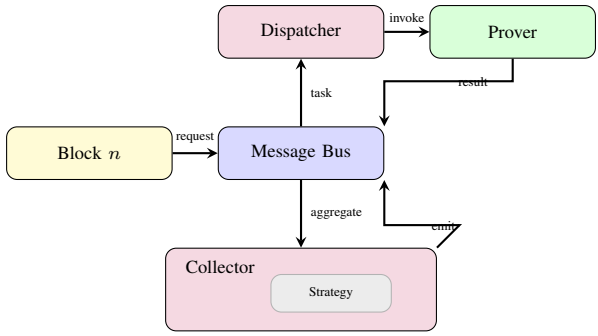

\subsection{Component Design}

\paragraph{Message Bus}
The message bus provides persistent priority queues with at-least-once delivery semantics.
Messages are persisted to durable storage before acknowledging receipt from producers, ensuring no task is lost due to failures; configurable replication factor ($n \geq 3$) provides fault tolerance against storage failures.
Each task type maps to a dedicated queue.
Dispatchers acknowledge messages only after: (1) prover execution completes successfully, (2) result is published to the output queue, and (3) output publication is confirmed.
This three-phase acknowledgment protocol ensures no work is lost even if a dispatcher crashes mid-execution.
Unacknowledged tasks are automatically redelivered after a configurable timeout $T_{\text{ack}}$; a maximum retry counter $R_{\text{max}}$ bounds redelivery attempts, with messages exceeding this threshold routed to a dead-letter queue for manual inspection.
Priority ordering ensures older blocks and failed-then-retried tasks are processed first, satisfying R1 (head-of-chain ordering).

Three-layer backpressure prevents resource exhaustion under heavy load:
(1) queue depth limits ($Q_{\text{max}}$ per queue; producers block when reached),
(2) worker flow control (dispatchers limit concurrent in-flight tasks to $F_{\text{max}}$),
and (3) optional per-producer rate limits.
Within a task type, all dispatchers join the same \emph{worker group} (exclusive delivery group), ensuring each message is delivered to exactly one dispatcher and preventing redundant computation.
The design is agnostic to the underlying message bus implementation; any system providing persistence, worker groups, and timeout-based redelivery suffices (e.g., NATS JetStream, RabbitMQ, Apache Kafka, Redis Streams).

\emph{Message Bus Resilience.}
Production message buses support clustered deployment with consensus-based replication for durability.
In our production deployment, we use NATS JetStream~\cite{nats} in a 3-node cluster across availability zones; the cluster tolerates single-node failures without message loss.
If the entire message bus becomes unavailable, dispatchers and collectors block on connection retry with exponential backoff.
Upon bus recovery, in-flight messages (those dispatched but not acknowledged) are automatically redelivered, ensuring no proving work is lost.
The system trades availability for consistency: during a bus partition, no new tasks are dispatched, but no duplicate or out-of-order proofs are produced.

\paragraph{Dispatcher}
A dispatcher is parameterized by: (1) an input queue identifier, (2) an output queue identifier, and (3) a prover binary path.
Upon receiving a task, the dispatcher optionally checks for an existing result with the same TaskID; if found, it acknowledges the input without re-execution.
Otherwise, it invokes the binary via a standardized interface, awaits completion, and publishes the result.
Redundant results from message redelivery are deduplicated by downstream collectors, ensuring at-most-once aggregation.
Multiple dispatchers can read from the same queue, enabling horizontal scaling (R5).
For long-running proofs, dispatchers emit periodic heartbeat signals (default: every 5 seconds) to the message bus, extending the acknowledgment deadline.
This prevents spurious task redelivery during legitimate multi-minute proof computations while preserving fast failure detection for genuinely stalled provers.
Dispatchers are stateless; they can be preempted and rescheduled without data loss.

\paragraph{Collector}
Collectors aggregate results from one or more input queues using a pluggable \emph{collection strategy}.
Each task message carries a JSON payload with application-specific fields (e.g., block number, task metadata); the message bus assigns a unique sequence ID used for deduplication.
Collectors extract a configurable grouping key (e.g., \texttt{block\_num}) from the payload to determine which messages belong to the same aggregation group; mismatched inputs are rejected and logged.
A deduplication cache keyed by message ID discards duplicate results from message redelivery, ensuring at-most-once aggregation.

Each incoming message is passed to the strategy's \texttt{collect()} method, which returns a \texttt{CollectorOperation}:
\begin{itemize}
    \item \texttt{InProgress}: accumulate the message; barrier not yet complete
    \item \texttt{Finished(groups)}: barrier complete; emit aggregated results and acknowledge all constituent messages
    \item \texttt{Postpone(msg, timeout)}: requeue the message for later processing
    \item \texttt{Skip(msg)}: discard the message without processing
\end{itemize}
For recursive proof composition, the strategy waits for $k$ subproofs sharing a common \texttt{GROUPING\_FIELD} (e.g., \texttt{block\_num}) before returning \texttt{Finished}.
The grouping field identifies which messages belong to the same aggregation barrier---for instance, all chunk proofs for block $n$ share \texttt{block\_num=$n$} and must be collected before emitting the aggregation task.
Collectors track expected input counts per group; when a timeout $T_{\text{collect}}$ expires with fewer than $k$ inputs received, the collector invokes the collection strategy, which may emit partial results if its completion conditions are met.
Collectors maintain soft state (aggregation buffers, pending message sets) that can be reconstructed from the message bus on restart; this design allows horizontal scaling with at-most-once aggregation semantics.

\paragraph{Prover Interface}
push0 provides two prover integration modes, balancing ease of integration against performance:

\emph{JSON Executor (Simple Integration).}
Any prover can be integrated via a file-based JSON interface: accept input from \texttt{--input-path}, write output to \texttt{--output-path}.
Provers with non-conforming CLIs require minimal gluing code---a shell wrapper that translates arguments and file paths.
For example, integrating a prover expecting \texttt{-i input.json -o output.json} requires minimal efforts.
Dispatchers invoke the binary via \texttt{std::process::Command}, capturing exit codes, output paths, and standard streams.
This mode imposes file I/O overhead but requires no changes to the orchestration layer.

\emph{Composer Executor (High-Performance Extension).}
For provers available as Rust libraries or when zero-copy data sharing is critical, the \emph{composer} interface enables native integration.
Provers implement a Rust trait specifying input/output types:
\begin{lstlisting}[basicstyle=\ttfamily\footnotesize]
pub trait Prover<Input, Output> {
    fn prove(&self, input: Input) -> Result<Output>;
}
\end{lstlisting}
The dispatcher instantiates the prover at startup and invokes \texttt{prove()} directly, eliminating file I/O and JSON serialization overhead.
This mode provides better performance and flexibility (e.g., in-memory witness manipulation) but requires prover-specific Rust code.

Both modes support long-running proofs via the dispatcher heartbeat mechanism described above.

Critically, the orchestration layer (dispatchers, collectors, queue topology, fault recovery) is fully prover-agnostic---integration requires adaptation only at the prover invocation boundary (gluing code for JSON executor, trait implementation for Composer executor).
This design enables heterogeneous zkVM orchestration, from CLI-based provers to library-integrated systems, satisfying R4.

\subsection{Workflow Specification}

Conceptually, a proving workflow forms a DAG where nodes are task types and edges are data dependencies.
push0 realizes this DAG \emph{implicitly} through queue topology and collection strategy configuration---there is no explicit DAG object.
Each stage connects via input/output queue names; the DAG structure emerges from the message flow.
For a typical zkEVM, the workflow might comprise:
\begin{enumerate}
    \item \textbf{Witness generation}: Extract execution trace from block.
    \item \textbf{Subproof generation}: Prove individual circuit segments in parallel.
    \item \textbf{Recursive aggregation}: Compose subproofs into a single succinct proof.
\end{enumerate}
The collector at each stage implements the appropriate collection strategy: e.g., ``aggregate all segment proofs for block $n$ before emitting aggregation task'' uses the Match strategy with \texttt{GROUPING\_FIELD=block\_num}.

\subsubsection{Production Pipeline Example}

Before describing push0's pipeline primitives, we ground the discussion in a concrete production deployment.
\Cref{fig:scroll-pipeline} illustrates a production ZK-rollup pipeline with two parallel tracks---a \emph{proposer flow} (which commits to batch/bundle boundaries using metadata) and a \emph{proving flow} (which generates cryptographic proofs)---synchronized by aggregator collectors.

\begin{figure}[htbp]
\centering
\scriptsize
\begin{tabular}{@{}p{2.8cm}p{3.2cm}l@{}}
\toprule
\textbf{Stage} & \textbf{In $\rightarrow$ Out} & \textbf{Flow} \\
\midrule
\multicolumn{3}{l}{\textit{Proposer Flow (metadata only)}} \\
Chunk Proposer & blocks $\rightarrow$ chunks & Collector \\
Batch Proposer & chunks $\rightarrow$ batches & Collector \\
Bundle Proposer & batches $\rightarrow$ bundles & Collector \\
\midrule
\multicolumn{3}{l}{\textit{Proving Flow (cryptographic)}} \\
Chunk Prover & chunks $\rightarrow$ chunk\_proofs & Dispatcher \\
Compress (2 stages) & chunk\_proofs $\rightarrow$ thin & Dispatcher \\
\midrule
\multicolumn{3}{l}{\textit{Synchronization (multi-input)}} \\
Batch Aggregator & thin + batches $\rightarrow$ agg & Multi-Input \\
Batch Prover & agg $\rightarrow$ batch\_proofs & Dispatcher \\
\bottomrule
\end{tabular}
\caption{Production ZK-rollup pipeline. \emph{Chunks} are fixed-size block ranges proven independently; \emph{batches} group chunks for recursive aggregation; \emph{bundles} group batches for L1 submission. Proposer flow commits to these boundaries using metadata only; proving flow generates proofs in parallel. Aggregators synchronize both flows before downstream proving.}
\label{fig:scroll-pipeline}
\end{figure}

In this pipeline, a \emph{chunk} is a fixed-size range of blocks (e.g., 10 blocks) proven independently.
A \emph{batch} groups multiple chunks for recursive proof aggregation.
A \emph{bundle} groups batches for final L1 submission.
The proposer flow commits to chunk/batch/bundle boundaries using block metadata alone, without waiting for cryptographic proofs.
The Batch Aggregator synchronizes both flows: it reads from both \texttt{thin\_proofs} (compressed chunk proofs) and \texttt{batches} (batch proposals), emitting downstream tasks only when both inputs for a given batch are available.

Critically, the Batch Proposer reads from \texttt{chunks} (proposed chunk metadata), \emph{not} from \texttt{chunk\_proofs}.
This decoupling means batch boundaries are committed before chunk proving completes---if chunk provers are slow or fail, the proposer flow continues unimpeded.
The aggregator buffers proposals until corresponding proofs arrive, bounding memory via timeout-based eviction.

\subsubsection{Pipeline Patterns}

push0 supports four fundamental pipeline patterns, illustrated in \Cref{fig:pipeline-patterns}.
These patterns compose to form complex proving workflows.

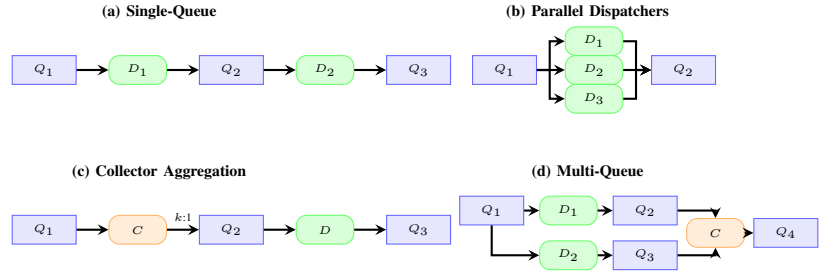
\begin{figure*}[htbp]
\centering
\begin{tikzpicture}[
    scale=0.85, transform shape,
    node distance=0.5cm,
    queue/.style={rectangle, draw=blue!60, fill=blue!10, minimum width=1.0cm, minimum height=0.45cm, font=\tiny},
    disp/.style={rectangle, rounded corners, draw=green!60, fill=green!15, minimum width=0.9cm, minimum height=0.45cm, font=\tiny},
    coll/.style={rectangle, rounded corners, draw=orange!60, fill=orange!15, minimum width=0.9cm, minimum height=0.45cm, font=\tiny},
    arr/.style={->, thick, >=stealth},
    label/.style={font=\scriptsize\bfseries}
]

\node[label] at (1.8, 0.9) {(a) Single-Queue};
\node[queue] (q1) at (0, 0) {$Q_1$};
\node[disp, right=of q1] (d1) {$D_1$};
\node[queue, right=of d1] (q2) {$Q_2$};
\node[disp, right=of q2] (d2) {$D_2$};
\node[queue, right=of d2] (q3) {$Q_3$};
\draw[arr] (q1) -- (d1);
\draw[arr] (d1) -- (q2);
\draw[arr] (q2) -- (d2);
\draw[arr] (d2) -- (q3);

\node[label] at (8.5, 0.9) {(b) Parallel Dispatchers};
\node[queue] (fo_in) at (7.2, 0) {$Q_1$};
\node[disp] (fo_d1) at (8.6, 0.45) {$D_1$};
\node[disp] (fo_d2) at (8.6, 0) {$D_2$};
\node[disp] (fo_d3) at (8.6, -0.45) {$D_3$};
\node[queue] (fo_out) at (10, 0) {$Q_2$};
\draw[arr] (fo_in.east) -- ++(0.2,0) |- (fo_d1.west);
\draw[arr] (fo_in.east) -- (fo_d2.west);
\draw[arr] (fo_in.east) -- ++(0.2,0) |- (fo_d3.west);
\draw[arr] (fo_d1.east) -| ++(0.2,-0.45) -- (fo_out.west);
\draw[arr] (fo_d2.east) -- (fo_out.west);
\draw[arr] (fo_d3.east) -| ++(0.2,0.45) -- (fo_out.west);

\node[label] at (1.8, -1.6) {(c) Collector Aggregation};
\node[queue] (fi_in) at (0, -2.5) {$Q_1$};
\node[coll, right=of fi_in] (fi_c) {$C$};
\node[queue, right=of fi_c] (fi_mid) {$Q_2$};
\node[disp, right=of fi_mid] (fi_d) {$D$};
\node[queue, right=of fi_d] (fi_out) {$Q_3$};
\draw[arr] (fi_in) -- (fi_c);
\draw[arr] (fi_c) -- node[above, font=\tiny] {$k$:1} (fi_mid);
\draw[arr] (fi_mid) -- (fi_d);
\draw[arr] (fi_d) -- (fi_out);

\node[label] at (8.5, -1.6) {(d) Multi-Queue};
\node[queue] (mq_q1) at (7, -2.2) {$Q_1$};
\node[disp] (mq_d1) at (8.2, -2.2) {$D_1$};
\node[queue] (mq_q2) at (9.4, -2.2) {$Q_2$};
\node[disp] (mq_d2) at (8.2, -2.9) {$D_2$};
\node[queue] (mq_q3) at (9.4, -2.9) {$Q_3$};
\node[coll] (mq_c) at (10.5, -2.55) {$C$};
\node[queue] (mq_out) at (11.6, -2.55) {$Q_4$};

\draw[arr] (mq_q1) -- (mq_d1);
\draw[arr] (mq_d1) -- (mq_q2);
\draw[arr] (mq_q1) |- (mq_d2);
\draw[arr] (mq_d2) -- (mq_q3);
\draw[arr] (mq_q2) -| (mq_c);
\draw[arr] (mq_q3) -| (mq_c);
\draw[arr] (mq_c) -- (mq_out);

\end{tikzpicture}
\caption{Pipeline patterns. Dispatchers ($D$, green) invoke provers; Collectors ($C$, orange) aggregate. (a) Sequential: $D_1 \to Q_2 \to D_2$. (b) Parallel: multiple dispatchers share input queue. (c) Fan-in: collector waits for $k$ messages before emitting one. (d) Multi-queue: collector reads from multiple queues, synchronizes independent flows.}
\label{fig:pipeline-patterns}
\end{figure*}

\paragraph{Single-Queue Pipeline}
The simplest pattern chains dispatchers via queues.
Each dispatcher consumes from one queue and produces to another:
\begin{lstlisting}[basicstyle=\ttfamily\scriptsize]
# Dispatcher D1: witness generation
INPUT_QUEUE: blocks
OUTPUT_QUEUE: traces

# Dispatcher D2: proving
INPUT_QUEUE: traces
OUTPUT_QUEUE: proofs
\end{lstlisting}
Tasks flow sequentially; no collector is needed when each input produces exactly one output.

\paragraph{Parallel Dispatchers}
Multiple dispatcher instances consume from a shared input queue, enabling horizontal scaling:
\begin{lstlisting}[basicstyle=\ttfamily\scriptsize]
# N dispatcher instances, same config
INPUT_QUEUE: chunk_tasks
OUTPUT_QUEUE: chunk_proofs
REPLICAS: 8
\end{lstlisting}
Each dispatcher pulls the next available task; the message bus ensures at-most-once delivery per task.

\paragraph{Collector Aggregation}
A collector waits for $k$ related messages before emitting:
\begin{lstlisting}[basicstyle=\ttfamily\scriptsize]
# Collector C: aggregate k proofs
INPUT_QUEUE: chunk_proofs
OUTPUT_QUEUE: aggregation_tasks
NUM_INPUTS: 8
GROUPING_FIELD: block_num
STRATEGY: Match
\end{lstlisting}
When 8 proofs with the same \texttt{block\_num} arrive, the collector bundles them and emits one message.

\paragraph{Multi-Queue Pipeline}
The most powerful pattern: a collector reads from multiple queues, synchronizing independent flows.
\Cref{fig:decoupled-example} shows two dispatcher chains feeding a single collector.

\begin{figure}[htbp]
\centering
\begin{tikzpicture}[
    scale=0.75, transform shape,
    node distance=0.35cm,
    queue/.style={rectangle, draw=blue!60, fill=blue!10, minimum width=0.9cm, minimum height=0.4cm, font=\tiny},
    disp/.style={rectangle, rounded corners, draw=green!60, fill=green!15, minimum width=0.8cm, minimum height=0.4cm, font=\tiny},
    coll/.style={rectangle, rounded corners, draw=orange!60, fill=orange!15, minimum width=0.8cm, minimum height=0.4cm, font=\tiny},
    arr/.style={->, thick, >=stealth},
    flowlabel/.style={font=\tiny\itshape, text=gray}
]

\node[queue] (q1) at (0, 0) {$Q_1$};

\node[flowlabel] at (1.5, 0.95) {Flow 1};
\node[disp] (d1) at (1.5, 0.5) {$D_1$};
\node[queue] (q2) at (2.8, 0.5) {$Q_2$};
\node[disp] (d2) at (4.1, 0.5) {$D_2$};
\node[queue] (q3) at (5.4, 0.5) {$Q_3$};

\node[flowlabel] at (1.5, -0.95) {Flow 2};
\node[disp] (d3) at (1.5, -0.5) {$D_3$};
\node[queue] (q4) at (2.8, -0.5) {$Q_4$};
\node[disp] (d4) at (4.1, -0.5) {$D_4$};
\node[queue] (q5) at (5.4, -0.5) {$Q_5$};

\node[coll] (c) at (6.5, 0) {$C$};
\node[queue] (q6) at (7.6, 0) {$Q_6$};

\draw[arr] (q1) |- (d1);
\draw[arr] (d1) -- (q2);
\draw[arr] (q2) -- (d2);
\draw[arr] (d2) -- (q3);

\draw[arr] (q1) |- (d3);
\draw[arr] (d3) -- (q4);
\draw[arr] (q4) -- (d4);
\draw[arr] (d4) -- (q5);

\draw[arr] (q3) -| (c);
\draw[arr] (q5) -| (c);
\draw[arr] (c) -- (q6);

\end{tikzpicture}
\caption{Multi-queue pipeline: input $Q_1$ feeds two independent dispatcher chains. Collector $C$ reads from both $Q_3$ and $Q_5$, emitting to $Q_6$ only when matching messages from both flows arrive. Flows may complete at different rates, so the collector buffers faster results until the slower flow catches up.}
\label{fig:decoupled-example}
\end{figure}
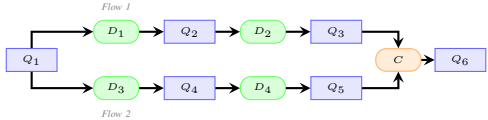

The collector reads from multiple input queues:
\begin{lstlisting}[basicstyle=\ttfamily\scriptsize]
# Collector C: synchronize two flows
INPUT_QUEUE: thin_proofs,batches
OUTPUT_QUEUE: aggregation_tasks
GROUPING_FIELD: block_num
NUM_INPUTS: 2
STRATEGY: Match
\end{lstlisting}

This pattern enables three key optimizations:
\begin{enumerate}
    \item \emph{Pipeline parallelism}: Flow 2 processes block $n+1$ while Flow 1 completes block $n$
    \item \emph{Flow isolation}: Slow dispatchers in one flow don't block the other
    \item \emph{Independent scaling}: Scale each flow's dispatchers separately based on load
\end{enumerate}

\subsubsection{Multi-Queue Pipelines}

As illustrated by the production pipeline (\Cref{fig:scroll-pipeline}), real-world ZK proving workflows involve complex multi-stage pipelines with fan-out and fan-in patterns.
A naive linear pipeline---where each stage blocks on the previous---creates unnecessary dependencies: if batch proposals must wait for chunk proofs, a slow prover delays not only proving but also proposal commitment.

push0 addresses this through \emph{decoupled parallel flows} synchronized by \emph{multi-input collectors}.
The key insight, demonstrated in the production pipeline, is separating the \emph{proposer flow} (which commits to chunk/batch/bundle boundaries) from the \emph{proving flow} (which generates cryptographic proofs).
Proposers can commit to batch structure based on chunk metadata alone, without waiting for chunk proofs.
An \emph{aggregator collector} later synchronizes proposals with their corresponding proofs, emitting downstream tasks only when both inputs are available.

This decoupling provides three benefits:
\begin{enumerate}
    \item \emph{Early commitment}: Batch and bundle boundaries are finalized before proofs complete, enabling pipeline parallelism.
    \item \emph{Fault isolation}: Proving pipeline failures do not block proposal progression; only aggregation stalls until proofs arrive.
    \item \emph{Independent scaling}: Proposer and prover resources scale independently based on their respective bottlenecks.
\end{enumerate}

push0 implements this pattern through \emph{multi-input collectors} that consume from multiple queues simultaneously.

\paragraph{Multi-Input Reading}
When a collector must synchronize outputs from independent pipeline stages---for example, matching batch proposals with their corresponding proofs, as in \Cref{fig:scroll-pipeline}---it reads from multiple input queues specified as comma-separated subjects:
\begin{lstlisting}[basicstyle=\ttfamily\scriptsize]
INPUT_QUEUE: thin_proofs,batches
ONE_CONSUMER_PER_SUBJECT: true
\end{lstlisting}
When \texttt{ONE\_CONSUMER\_PER\_SUBJECT} is enabled, the collector creates a dedicated consumer for each input subject, preventing starvation when one queue has higher throughput than another.
The collection strategy then determines when inputs from all queues have been matched (e.g., both a batch proposal and its corresponding proofs have arrived for the same block).

\paragraph{Collection Strategies}
push0 provides pluggable collection strategies that define barrier completion conditions:

\emph{Match Strategy.}
Aggregates messages sharing a common grouping field (e.g., block number) until a target count is reached.
For each message, the strategy extracts the grouping key and accumulates partial results:
\begin{lstlisting}[basicstyle=\ttfamily\scriptsize]
NUM_INPUTS: 4
GROUPING_FIELD: block_num
STRATEGY: Match
\end{lstlisting}
When 4 messages with the same \texttt{block\_num} arrive, the collector emits an aggregated result and acknowledges all constituent messages atomically.

\emph{Sequential Strategy.}
Groups messages by sequential ranges rather than exact matches.
For batch aggregation where blocks $[n, n+k)$ form a single proving unit:
\begin{lstlisting}[basicstyle=\ttfamily\scriptsize]
NUM_INPUTS: 10
GROUPING_FIELD: block_num
STRATEGY: Sequential
\end{lstlisting}
The strategy computes \texttt{group\_id = block\_num / NUM\_INPUTS}, collecting all blocks in the same batch window.

\emph{Custom Strategies.}
Complex workflows require domain-specific aggregation logic.
The strategy trait allows custom implementations:
\begin{lstlisting}[basicstyle=\ttfamily\scriptsize]
trait Strategy<T: Message> {
  fn collect(&mut self, task: Payload, msg: T)
    -> Result<CollectorOperation<T>>;
  fn collector_id(&self, task: Value,
    num_collectors: u8) -> u8;
}
\end{lstlisting}
The \texttt{collector\_id} method implements partition-affine routing (Equation~\ref{eq:partition-routing}), ensuring messages for the same aggregation group route to a single collector instance.

\paragraph{Timeout-Based Collection}
For variable-size aggregation (e.g., batches with dynamic chunk counts), collectors support timeout-triggered emission:
\begin{lstlisting}[basicstyle=\ttfamily\scriptsize]
COLLECT_TIMEOUT_PERIOD_MILLIS: 1000
\end{lstlisting}
If no new message arrives within the timeout period, the collector emits whatever partial results have accumulated.
This prevents indefinite blocking when the expected input count is unknown or when upstream failures prevent complete batches.

\subsection{Scheduling Model and Correctness}

We formalize the scheduling model underlying push0 and establish correctness properties for task ordering, barrier synchronization, and scaling operations.

\paragraph{System Model}
We consider an asynchronous message-passing system with $D$ dispatchers and $C$ collectors.
Processes may fail by crashing (fail-stop model); we assume the message bus provides reliable delivery with at-least-once semantics.
Tasks are partially ordered: tasks for block $n$ must complete before block $n+1$ proofs can be submitted, but subtasks within a block execute concurrently.

\paragraph{Task Scheduling}
Dispatchers consume tasks from priority queues ordered by $(\text{block\_num}, \text{retry\_count}, \text{enqueue\_time})$.
This lexicographic ordering provides:
\begin{itemize}
    \item \emph{Head-of-chain priority}: Lower block numbers always dequeue first, ensuring oldest blocks complete before newer ones.
    \item \emph{Retry escalation}: Failed tasks receive elevated priority proportional to retry count, preventing starvation of problematic blocks.
    \item \emph{FIFO tiebreaking}: Among equal-priority tasks, oldest-first ordering ensures fairness.
\end{itemize}

\begin{theorem}[Starvation Freedom]
Under the priority ordering above, every enqueued task eventually reaches the head of the queue, assuming finite task arrivals and bounded retry counts.
\end{theorem}
\begin{proof}[Proof sketch]
Tasks with lower block numbers have strictly higher priority; since block numbers are bounded from below and increase monotonically, tasks cannot be indefinitely preempted.
Retry escalation ensures failed tasks eventually exceed the priority of new arrivals.
\end{proof}

\subsubsection{Partition-Affine Routing}

Proof aggregation requires collecting $k$ partial results before emitting successor tasks---a \emph{distributed barrier} synchronization problem~\cite{hensgen1988barrier}.
When multiple collectors share a single input queue, barrier semantics are violated.

Consider collectors $\{C_0, C_1\}$ aggregating $k=2$ proofs per block:
\begin{enumerate}
    \item Dispatcher publishes $\pi_{\square}$ for block $n$; message bus delivers to $C_0$
    \item Dispatcher publishes $\pi_{\bigcirc}$ for block $n$; message bus delivers to $C_1$
    \item $C_0$ holds $\pi_{\square}$, awaits second input; $C_1$ holds $\pi_{\bigcirc}$, awaits second input
    \item \textbf{Livelock}: Both collectors wait indefinitely; neither can complete the barrier
\end{enumerate}

This \emph{barrier fragmentation} problem---where inputs for a single barrier are scattered across uncoordinated processes---is well-understood in stream processing~\cite{carbone2015apache,kafka_streams}.
We formalize its impact for ZK proof aggregation:

\noindent\textbf{Observation (Barrier Fragmentation).}
\label{thm:fragmentation}
With $C \geq 2$ collectors and $k \geq 2$ inputs per barrier under uniform random assignment, each input independently lands on one of $C$ collectors.
All $k$ must reach the same collector for the barrier to complete: fixing the first input's destination, each subsequent input matches with probability $1/C$, giving completion probability $C^{-(k-1)}$.
For typical parameters ($C=4$, $k=8$), this is $4^{-7} \approx 0.006\%$; incomplete barriers consume collector memory indefinitely.

push0 solves barrier fragmentation through \emph{partition-affine routing}: all messages for a given aggregation group route to a single collector, ensuring barrier inputs are co-located.
The routing scheme is a degenerate case of consistent hashing~\cite{karger1997consistent} optimized for the ZK proving workload.
Each collector $c_i$ reads exclusively from partition $P_i$.
For message $m$ in aggregation group $g$ (typically $g = \text{block\_num}$):
\begin{equation}
\text{partition}(m) = h(g) \mod C
\label{eq:partition-routing}
\end{equation}
where $h$ is the identity function for sequential block numbers or a cryptographic hash for arbitrary group identifiers.
Simple modulo suffices because ZK proving has favorable properties: (1) block numbers are sequential with uniform distribution modulo $C$; (2) collector capacities are homogeneous; (3) the collector count $C$ changes infrequently.

\begin{theorem}[Barrier Locality]
\label{thm:locality}
Under partition-affine routing, all messages for aggregation group $g$ are delivered to a single collector $c_{h(g) \mod C}$.
\end{theorem}
\begin{proof}
For any two messages $m_1, m_2$ with $\text{group}(m_1) = \text{group}(m_2) = g$:
$\text{partition}(m_1) = h(g) \mod C = \text{partition}(m_2)$.
Since each partition maps to exactly one collector, both messages route to the same collector.
\end{proof}

\begin{corollary}[Barrier Completion]
If all $k$ partial results for group $g$ are published, the collector $c_{h(g) \mod C}$ eventually receives all $k$ inputs and completes the barrier.
\end{corollary}

With $B$ blocks distributed across $C$ collectors, expected load per collector is $B/C$.
For sequential block numbers, maximum load imbalance is bounded by $\lceil B/C \rceil - \lfloor B/C \rfloor \leq 1$ block.
For arbitrary keys, 2-universal hashing provides similar balance~\cite{karger1997consistent}.

\subsubsection{Scale-Down and Routing Continuity}

Reducing collector count from $C$ to $C' < C$ orphans partitions $P_{C'}, \ldots, P_{C-1}$.
Messages in orphaned partitions become unreachable, violating barrier completion.
push0 addresses this through three mechanisms:

\emph{Graceful Drain.}
Collectors entering shutdown drain in-flight barriers:
\begin{enumerate}
    \item Cease accepting new messages; signal ``draining'' to orchestrator
    \item Complete pending barriers within grace period $T_{\text{grace}} = 2 \times T_{\text{barrier}}$
    \item Publish completed results; acknowledge processed messages
    \item Terminate after drain completes or grace period expires
\end{enumerate}

\emph{Partition Takeover.}
Surviving collectors assume orphaned partitions.
When scaling from $C$ to $C'$, collector $c_i$ (where $i < C'$) reads from:
\begin{equation}
\text{Partitions}(c_i) = \{ P_j : j \mod C' = i \}
\label{eq:partition-takeover}
\end{equation}

This preserves a critical invariant: for any group $g$ previously routed to partition $P_j$ where $j \geq C'$, new messages route to $P_{h(g) \mod C'}$, and $P_j$ is assumed by collector $c_{j \mod C'}$---the same collector receiving new messages for $g$.

\begin{theorem}[Routing Continuity]
Under partition takeover, in-flight messages for group $g$ and new messages for $g$ are processed by the same collector.
\end{theorem}
\begin{proof}
Let $g$ have original partition $j = h(g) \mod C$ where $j \geq C'$.
New partition: $h(g) \mod C' = j \mod C'$ (since $h(g) \equiv j \pmod{C}$; for $j \geq C'$, $j \mod C' \in \{0, \ldots, C'-1\}$).
Collector $c_{j \mod C'}$ assumes $P_j$ by Equation~\ref{eq:partition-takeover} and receives new messages for partition $j \mod C'$.
Since both map to $c_{j \mod C'}$, continuity holds.
\end{proof}

\emph{Unacknowledged Message Redelivery.}
If a collector terminates abruptly without completing the graceful drain (e.g., crash or forced kill), messages it had consumed but not acknowledged are automatically redelivered by the message bus after the ACK timeout $T_{\text{ack}}$ expires.
The new partition assignment routes these redelivered messages to the surviving collector responsible for the orphaned partition, ensuring no barrier inputs are lost even under ungraceful termination.

\subsubsection{Speculative Redundancy}

For latency-critical deployments, push0 supports intentional redundancy via $k$-of-$n$ dispatch: the same task is sent to $n$ provers, accepting the first $k$ valid results.
This converts compute cost into tail latency reduction.
Dispatchers publish to $n$ prover-specific queues; collectors accept the first valid result per TaskID and deduplicate subsequent arrivals.
Outstanding redundant tasks naturally expire via acknowledgment timeout, avoiding explicit cancellation complexity.

\section{Evaluation}
\label{sec:evaluation}

\subsection{Experimental Setup}

\paragraph{Test Environment}
Primary experiments run on a 3-node Kubernetes cluster (AWS t3.medium instances, 2 vCPUs, 4\,GB RAM each) with Synadia NGS cloud-hosted NATS to validate production deployment characteristics.
This configuration represents realistic distributed proving infrastructure where dispatchers, collectors, and message bus span separate networks.
Supplementary controlled experiments use Docker Compose on a single machine (Apple M3 Max, 36\,GB RAM) with local NATS JetStream to isolate specific orchestration behaviors without network variance.

\paragraph{Simulated Provers}
Controlled experiments (\S\ref{subsec:latency}--\S\ref{subsec:scalability}) use an echo executor that returns immediately or sleeps for a configurable duration (0--5 seconds).
This methodology isolates pure orchestration overhead: with 0\,ms provers, all measured latency is coordination cost; with multi-second provers, we observe how overhead amortizes under realistic workloads.
Real provers (zkEVM circuits) typically require 1--30 seconds per proof; our simulated latencies span this range.
Production deployment with SP1 GPU provers (\S\ref{subsec:production-validation}) provides ecological validity, demonstrating that orchestration behavior observed in controlled settings translates to real zkVM workloads.

\paragraph{Workload Generation}
Tasks are injected via Python scripts that publish JSON messages to NATS.
Each task contains a block number, task ID, and data type.
For scalability tests, we scale task count proportionally with dispatcher count ($N = 10 \times D$) to ensure even work distribution across consumer group members.

\paragraph{Metrics Collection}
Timing is measured from task publication to result receipt using Python's \texttt{time.time()}.
NATS stream state (message counts) provides completion verification.
Memory consumption is measured via Docker's \texttt{stats} API.

\subsection{Implementation}
We implemented push0 in Rust, comprising 4,746 lines of code for core orchestration logic (excluding tests, comments, and generated code).
The implementation is available at \url{https://github.com/zircuit-labs/push0-orchestrator}.
\Cref{tab:loc-breakdown} shows the per-component breakdown.

\begin{table}[htbp]
\caption{Code Size by Component}
\begin{center}
\begin{tabular}{|l|r|l|}
\hline
\textbf{Component} & \textbf{LOC} & \textbf{Purpose} \\ \hline
message\_policy & 1,622 & NATS JetStream integration \\ \hline
dispatcher & 644 & Task dispatch, prover invocation \\ \hline
config & 506 & Configuration parsing \\ \hline
collection\_strategy & 430 & Barrier strategies (Match, Seq) \\ \hline
collector & 412 & Result aggregation \\ \hline
storage & 412 & State persistence \\ \hline
traits & 280 & Core abstractions \\ \hline
executor & 224 & Prover executors (JSON, Composer) \\ \hline
traces & 216 & OpenTelemetry integration \\ \hline
\end{tabular}
\label{tab:loc-breakdown}
\end{center}
\end{table}
The message bus uses NATS JetStream~\cite{nats}; observability integrates OpenTelemetry for distributed tracing and Prometheus for metrics export.
A reference implementation demonstrating orchestration across heterogeneous zkVMs is publicly available~\cite{push0_demo2025}.

The architecture supports multi-network orchestration spanning bare-metal Kubernetes clusters, AWS managed infrastructure, and GPU accelerator pools.
Each infrastructure type connects to the shared message bus; dispatchers route tasks to appropriate hardware based on task requirements (e.g., GPU-intensive MSM operations to GPU clusters, CPU-bound witness generation to general-purpose nodes).

\subsection{Generalizability}
To validate orchestration-layer prover agnosticism, we deployed push0 across three distinct workflows demonstrating both integration modes:

\begin{enumerate}
    \item \textbf{Production rollup (JSON executor)}: A proprietary zkEVM with multi-stage proving (witness, chunked proofs, recursive aggregation), using CLI-based provers with shell wrappers to adapt argument formats.
    \item \textbf{Ethproofs demonstration (JSON executor)}: The Ethproofs L1 real-time proving pipeline~\cite{ethproofs2025}, orchestrating block witness generation and proof submission within slot timing constraints using file-based JSON interfaces.
    \item \textbf{SP1 zkVM pipeline (Composer executor)}: Succinct's SP1 RISC-V prover~\cite{sp1}, integrated via native Rust trait implementation.
    SP1's library-based architecture enables the high-performance Composer mode with zero-copy data sharing.
\end{enumerate}

Each workflow required different orchestration configurations reflecting its specific proving pipeline: the production zkEVM uses a four-stage queue topology (witness, chunk, aggregation, final proof), Ethproofs uses a two-stage pipeline (witness generation, proof submission), and SP1 uses pipeline-specific stages determined by the RISC-V circuit structure.
Despite these differences in queue topology and collection strategies, all workflows use the same orchestration software---dispatchers, collectors, message routing, and failure recovery logic remain identical.
Integration requires only prover-invocation-boundary adaptation: minimal JSON gluing code for CLI-based provers, Rust trait implementations for library-based provers.
This demonstrates that orchestration logic is fully prover-agnostic, with proving-pipeline-specific complexity isolated to configuration and integration boundaries.

\subsection{Latency Overhead}
\label{subsec:latency}
We isolate orchestration overhead by measuring the time from task publication to dispatcher completion, computed as \texttt{completed\_at\_ns - publish\_timestamp\_ns}.
Prover execution time is excluded, as it dominates end-to-end latency and varies by orders of magnitude across zkVMs.
All experiments use 10 dispatcher instances processing 1,000 tasks at each injection rate, repeated 10 times to capture mean and standard deviation.

\Cref{tab:latency-rates} reports measurements from the production Kubernetes cluster across varying task injection rates using an instant echo prover (0ms proving time) to isolate pure orchestration overhead.
The cluster achieves consistent sub-10\,ms median latency (P50: 5.2--5.3\,ms, P99: 8.0--8.5\,ms) across all injection rates, confirming that push0 overhead is negligible (0.05--0.1\%) relative to typical proof computation times (seconds to minutes).
Supplementary measurements from Docker Compose with local NATS show comparable performance (P50: 2.7--10\,ms), validating that the cluster results are limited by orchestration logic rather than deployment infrastructure.

\begin{table}[htbp]
\caption{Orchestration Latency at Varying Injection Rates (n=10)}
\begin{center}
\begin{tabular}{|r|r|r|}
\hline
\textbf{Rate} & \textbf{P50} & \textbf{P99} \\
\textbf{(tasks/s)} & \textbf{(ms)} & \textbf{(ms)} \\ \hline
\multicolumn{3}{|c|}{\textit{Production K8s + Synadia NGS (Primary)}} \\ \hline
10  & $5.3 \pm 0.1$ & $8.4 \pm 0.3$ \\ \hline
50  & $5.2 \pm 0.1$ & $8.0 \pm 0.2$ \\ \hline
100 & $5.2 \pm 0.1$ & $8.5 \pm 0.4$ \\ \hline
\multicolumn{3}{|c|}{\textit{Docker + Local NATS (Validation)}} \\ \hline
10  & $10.0 \pm 1.0$ & $18.0 \pm 0.4$ \\ \hline
50  & $4.2 \pm 2.0$  & $11.0 \pm 6.7$ \\ \hline
100 & $2.7 \pm 0.2$  & $4.9 \pm 0.2$ \\ \hline
\end{tabular}
\label{tab:latency-rates}
\end{center}
\end{table}

The production cluster maintains stable 5ms median latency regardless of injection rate, demonstrating that cloud-hosted NATS introduces minimal overhead.
Local deployment shows more variance (2.7--10\,ms P50) due to interference from other Docker processes, but validates the same sub-10\,ms orchestration overhead.

\paragraph{Resource Usage}
Peak memory consumption remains bounded: dispatchers use 8--14\,MB RSS, collectors use 8\,MB, and the NATS message bus uses 30\,MB.
The Prometheus \texttt{/metrics} endpoint imposes negligible monitoring overhead: under loads of 100--1,000 concurrent tasks, scrape latency averages 2.2\,ms (P99: 18.6\,ms), representing $<$0.25\% overhead per scrape.

\Cref{fig:latency-cdf} presents orchestration overhead distributions from supplementary Docker Compose experiments (10 repeated runs per rate).
These controlled measurements isolate orchestration behavior without network variance, showing that 50 and 100 tasks/sec achieve consistent 3--6\,ms median overhead with tight distributions.
The production Kubernetes cluster exhibits similar patterns with slightly higher but more stable latencies (5.2\,ms median across all rates), as shown in \Cref{tab:latency-rates}.

\begin{figure}[htbp]
    \centering
    \includegraphics[width=\columnwidth]{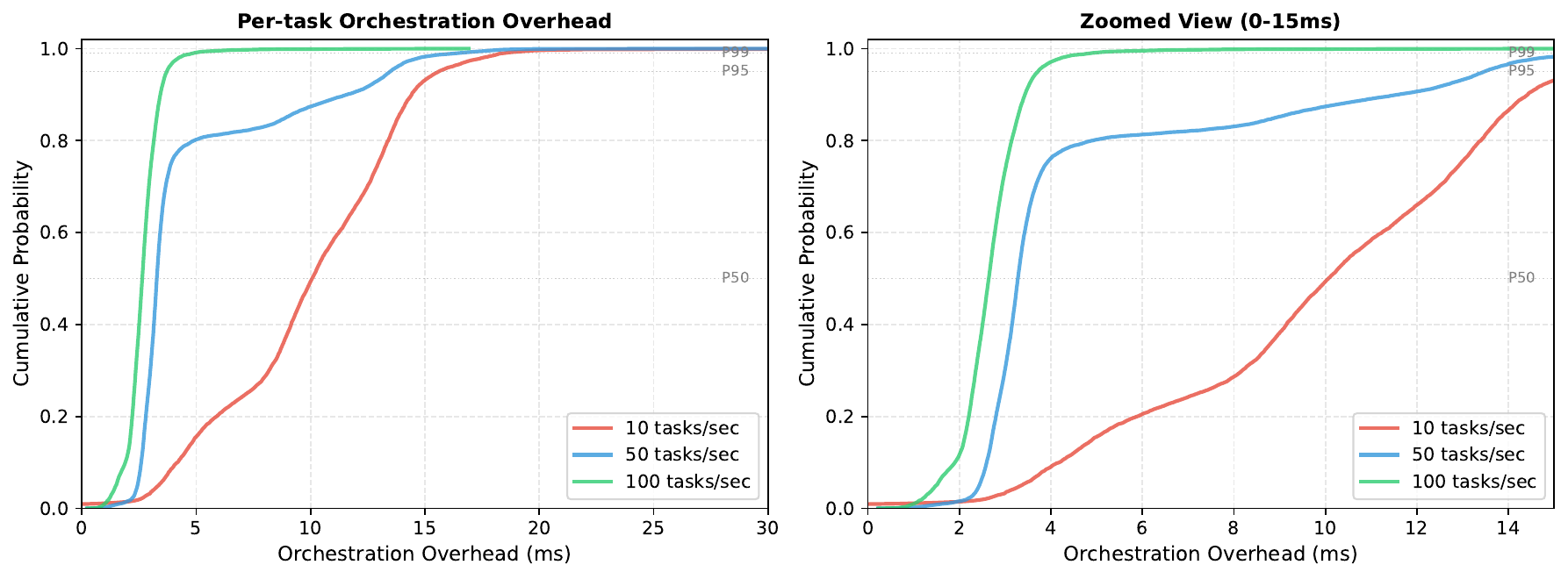}
    \caption{Orchestration overhead CDFs from controlled Docker deployment (n=10 runs per rate). Left: full range (0--30\,ms) shows all rates achieve sub-20\,ms P99. Right: zoomed view (0--15\,ms) highlights 50 and 100 tasks/sec maintain 3--6\,ms overhead. Production cluster results (not shown) exhibit similar distributions with 5.2\,ms median (\Cref{tab:latency-rates}).}
    \label{fig:latency-cdf}
\end{figure}

OpenTelemetry instrumentation (optional compile-time feature) adds typical tracing overhead of $<$1\% based on similar distributed systems~\cite{kaldor2017canopy}.


\subsection{Scalability}
\label{subsec:scalability}
We evaluate horizontal scaling by deploying 1--32 dispatcher instances reading from the same NATS stream.
To understand scaling behavior under different workloads, we test with simulated prover latencies from 0\,ms (stress test) to 5\,s (realistic).

\paragraph{Deployment Setup}
Scalability experiments run on a 3-node Kubernetes cluster (AWS t3.medium: 2 vCPUs, 4\,GB RAM per node) with Synadia NGS cloud-hosted NATS.
This configuration represents production ZK-rollup infrastructure where provers are distributed across disjoint networks and the message bus operates as a managed cloud service.
Each dispatcher is allocated 512\,MiB memory and 500m CPU limits; measured resource usage shows dispatchers consume 15--55\,mCores CPU and 8--14\,MB memory per instance (11\% CPU, 2.6\% memory utilization).

Deploying NATS outside the cluster via Synadia NGS forces messages to travel outside the cluster and back, capturing the network overhead inherent to geographically distributed proving infrastructure.
This approximates production scenarios where dispatchers coordinate provers on separate networks.

Supplementary experiments use Docker Compose with local NATS JetStream for controlled measurement of specific orchestration behaviors (e.g., Amdahl's Law stress tests with 0ms provers) where network variance would obscure the isolated effect.

\paragraph{Methodology}
Scaling efficiency is measured as $\eta = S/D$ where $S$ is speedup and $D$ is dispatcher count---this isolates parallelization loss from constant per-task overhead.
We scale task count proportionally with dispatchers ($N = 10 \times D$) to ensure even work distribution.
With fixed task counts, statistical variance in NATS consumer group distribution causes some dispatchers to receive disproportionately few tasks, artificially depressing efficiency.
Proportional scaling ensures each dispatcher processes a consistent workload, isolating scaling behavior from task distribution artifacts.

Under this methodology, wall-clock time remains constant (each dispatcher processes 10 tasks) while throughput increases linearly with dispatcher count.
Speedup is therefore measured as the ratio of throughput at $D$ dispatchers to throughput at 1 dispatcher: $S = \text{Throughput}_D / \text{Throughput}_1$.
Perfect linear scaling yields 100\% efficiency ($\eta = S/D = 1$).

\paragraph{Scaling with Realistic Prover Latency}
\Cref{tab:scalability-realistic} shows scaling on the production Kubernetes cluster with 1-second simulated provers---representative of lightweight ZK circuits.
Efficiency remains at 99\% across all dispatcher counts, demonstrating near-perfect linear speedup even with cloud-hosted NATS and distributed infrastructure.
Note that ``Speedup'' is a dimensionless multiplier relative to 1D baseline (e.g., 31.7$\times$ = 29.84/0.98), not absolute throughput.

\begin{table}[htbp]
\caption{Horizontal Scalability (10 tasks/dispatcher, 1s sleep prover, K8s + Synadia NGS)}
\begin{center}
\begin{tabular}{|r|r|r|r|r|}
\hline
\textbf{Dispatchers} & \textbf{Time (s)} & \textbf{Throughput} & \textbf{Speedup} & \textbf{Efficiency} \\ \hline
1  & 10.2 & 0.98/s  & 1.0$\times$  & 100\% \\ \hline
2  & 10.7 & 1.87/s  & 2.0$\times$  & 99\% \\ \hline
4  & 10.7 & 3.74/s  & 4.0$\times$  & 99\% \\ \hline
8  & 10.7 & 7.47/s  & 7.9$\times$  & 99\% \\ \hline
16 & 10.7 & 14.9/s  & 15.9$\times$ & 99\% \\ \hline
32 & 10.7 & 29.8/s  & 31.7$\times$ & 99\% \\ \hline
\end{tabular}
\label{tab:scalability-realistic}
\end{center}
\end{table}

\paragraph{Scaling with Production-Representative Latency}
\Cref{tab:scalability-production} shows scaling on the Kubernetes cluster with 5-second simulated provers---typical for zkEVM circuits.
Efficiency reaches 100\% at all scales, demonstrating near-perfect scaling.
This validates that orchestration overhead is negligible for production proving workloads.

\begin{table}[htbp]
\caption{Horizontal Scalability (10 tasks/dispatcher, 5s sleep prover, K8s + Synadia NGS)}
\begin{center}
\begin{tabular}{|r|r|r|r|r|}
\hline
\textbf{Dispatchers} & \textbf{Time (s)} & \textbf{Throughput} & \textbf{Speedup} & \textbf{Efficiency} \\ \hline
1  & 50.0 & 0.20/s & 1.0$\times$  & 100\% \\ \hline
2  & 50.0 & 0.40/s & 2.0$\times$  & 100\% \\ \hline
4  & 50.0 & 0.79/s & 4.0$\times$  & 100\% \\ \hline
8  & 50.0 & 1.58/s & 8.0$\times$  & 100\% \\ \hline
16 & 50.0 & 3.15/s & 16.0$\times$ & 100\% \\ \hline
32 & 50.0 & 6.31/s & 32.0$\times$ & 100\% \\ \hline
\end{tabular}
\label{tab:scalability-production}
\end{center}
\end{table}

\paragraph{Efficiency Across All Configurations}
\Cref{tab:scaling-efficiency-matrix} summarizes scaling efficiency ($S/D$) across prover latencies on the Kubernetes cluster.
Efficiency improves with prover latency because longer tasks amortize per-dispatch coordination---a direct consequence of Amdahl's Law.

\begin{table}[htbp]
\caption{Scaling Efficiency (\%) vs.\ Prover Latency (10 tasks/dispatcher, sleep provers, K8s + Synadia NGS)}
\begin{center}
\begin{tabular}{|r|r|r|r|r|r|r|}
\hline
\textbf{Prover} & \multicolumn{6}{c|}{\textbf{Dispatcher Count}} \\
\cline{2-7}
\textbf{Latency} & \textbf{1} & \textbf{2} & \textbf{4} & \textbf{8} & \textbf{16} & \textbf{32} \\ \hline
0\,ms    & 100 & 99 & 98 & 95 & 84 & 14 \\ \hline
100\,ms  & 100 & 100 & 100 & 99 & 97 & 55 \\ \hline
1\,s     & 100 & 99 & 99 & 99 & 99 & 99 \\ \hline
5\,s     & 100 & 100 & 100 & 100 & 100 & 100 \\ \hline
\end{tabular}
\label{tab:scaling-efficiency-matrix}
\end{center}
\end{table}

The 0\,ms row represents a stress test where prover work is eliminated entirely, isolating pure coordination overhead including cloud NATS round-trips.
By Amdahl's Law, speedup is bounded by $S = 1/(f + (1-f)/D)$ where $f$ is the serial fraction.
With instant provers, the serial component (cloud messaging, consumer group coordination) dominates; at 32 dispatchers, efficiency drops to 14\%.
This is expected and acceptable: production zkEVM provers take 1--10+ seconds, where the 5\,ms coordination overhead represents $<$0.1\% of task time.
Even with 100\,ms provers (10$\times$ faster than typical circuits), efficiency remains 55\% at 32 dispatchers.
For realistic workloads (1\,s+), scaling achieves 99--100\% efficiency, demonstrating near-perfect scaling.

\paragraph{Controlled Validation with Local NATS}
Supplementary Docker experiments with local NATS (eliminating network latency) achieve 29\% efficiency for 0ms provers at 32 dispatchers---2$\times$ higher than the cluster's 14\%, confirming that cloud messaging overhead dominates in the pathological instant-prover scenario.
For realistic workloads (1s+), both deployments achieve 99--100\% efficiency, validating that cluster results are representative of production infrastructure.


\paragraph{Partition Affinity Validation}
We empirically validate the partition affinity requirement formalized in Theorem~\ref{thm:fragmentation}.
With 4 collectors and $k=4$ barrier inputs per block, we submit 100 barriers (400 messages total) under two routing strategies.

\textbf{Round-robin routing} (naive load balancing): Collectors share a single input queue as a consumer group.
Each collector receives approximately $1/k$ of each barrier's inputs.
Result: \textbf{5/100 barriers completed (5\%)}.
The theoretical completion probability is $(1/k)^{k-1} = (1/4)^3 = 1.56\%$---the chance all $k$ inputs randomly land on the same collector.
Our observed 5\% is consistent with this expectation (binomial variance).

\textbf{Partition-affine routing}: Each collector has a dedicated input queue; messages are routed by $\text{block\_num} \mod k$.
All inputs for a given block are deterministically sent to the same collector.
Result: \textbf{100/100 barriers completed (100\%)}.

This quantifies the impact of the well-known partition affinity principle in the ZK aggregation context: naive load balancing causes near-total barrier fragmentation (95\% failure rate), while partition-affine routing achieves 100\% completion.

\subsection{Multi-Queue Pipeline Evaluation}
\label{sec:multiqueue-eval}

We empirically validate the multi-queue architecture's benefits (\S\ref{sec:design}) through a \emph{heterogeneous join-synchronization experiment} that demonstrates fault isolation and early commitment under adversarial conditions.

\paragraph{Experiment Design}
We configure two parallel tracks feeding a join collector:
\begin{itemize}
    \item \emph{Track X (Proposer Flow)}: 1 metadata task per block with 10\,ms latency---representing fast batch boundary commitment.
    \item \emph{Track Y (Prover Flow)}: 8 proof tasks per block with 300\,ms latency, processed by 4 parallel dispatchers---representing parallelizable cryptographic work.
\end{itemize}
The collector requires all 9 inputs (1 proposal + 8 proofs) before emitting a finalized batch.
To stress-test fault isolation, we inject a \emph{poison task}: at block 5, one prover task incurs a 15-second delay (50$\times$ normal latency), simulating a catastrophic prover stall such as GPU memory exhaustion or complex witness generation.

We compare two configurations:
\begin{enumerate}
    \item \textbf{Multi-queue (push0)}: Proposer and prover flows operate on independent queues; the collector synchronizes both flows.
    \item \textbf{Linear baseline}: All tasks flow through a single queue; each block's proposal cannot begin until the previous block's proofs complete.
\end{enumerate}

\paragraph{Results}
\Cref{tab:multiqueue-results} summarizes the experiment over 30 blocks.

\begin{table}[htbp]
\caption{Multi-Queue vs.\ Linear Pipeline under Prover Stall}
\begin{center}
\begin{tabular}{|l|r|r|}
\hline
\textbf{Metric} & \textbf{Multi-Queue} & \textbf{Linear} \\ \hline
Total Time (s) & 23.9 & 101.5 \\ \hline
Blocks Completed & 30 & 30 \\ \hline
Speedup & \multicolumn{2}{c|}{4.2$\times$} \\ \hline
Max Finality Skew (blocks) & 25 & 0 (blocked) \\ \hline
Proposer Throughput (blocks/s) & 15.9 & N/A \\ \hline
\end{tabular}
\label{tab:multiqueue-results}
\end{center}
\end{table}

The multi-queue architecture completed all 30 blocks in 23.9\,s---a \textbf{4.2$\times$ speedup} over the linear baseline (101.5\,s).
During the 15-second prover stall, the proposer flow continued unimpeded: the \emph{finality skew} (distance between committed proposals and completed proofs) reached 25 blocks.
This demonstrates early commitment: batch boundaries were finalized for blocks 6--30 while block 5's proofs remained pending.

The linear pipeline, by contrast, was \emph{completely blocked} during the stall---no forward progress on any block until the poison task resolved.
The entire 15-second delay propagated directly to end-to-end latency.

\paragraph{Implications}
These results validate the three benefits claimed in \S\ref{sec:design}:
\begin{enumerate}
    \item \emph{Early commitment}: Proposers committed batch boundaries for 25 blocks ahead of proof completion, enabling parallel pipeline execution.
    \item \emph{Fault isolation}: A catastrophic prover delay did not affect proposer throughput; only final aggregation stalled until proofs arrived.
    \item \emph{Independent scaling}: The 4 parallel prover dispatchers processed 8 proofs per block concurrently, while the single proposer handled metadata---each scaled to its bottleneck.
\end{enumerate}

For production ZK-rollups, this architecture ensures transaction ordering finality (R1) progresses even when proving infrastructure experiences delays or failures.

\paragraph{Collector Buffer Memory}
We measured collector memory under varying arrival skew---the ratio between fast proposals (Track X) and slow proofs (Track Y).
With skew ratios of 30$\times$, 50$\times$, and 100$\times$ (proposals arriving 30--100$\times$ faster than proofs), peak collector memory remained bounded at 3.7--7.7\,MiB.
This demonstrates efficient buffering: the collector accumulates proposals while awaiting proofs without unbounded memory growth.

\subsection{Fault Tolerance}
We evaluate fault tolerance by injecting dispatcher crashes during task processing and measuring task loss and recovery time.

\paragraph{Dispatcher Crash Recovery}
We ran a chaos test killing random dispatcher processes at 10\% task completion intervals.
With 2 dispatchers processing 50 tasks (1-second simulated proving time), we injected 2 failures.
\textbf{Zero tasks were lost}; all tasks completed successfully.
Recovery time when a healthy dispatcher was available averaged 5--6\,ms---the surviving dispatcher immediately picked up redelivered tasks.

The message bus redelivers unacknowledged tasks after the configured ACK timeout (default 30s).
When all dispatchers fail simultaneously, recovery time equals the ACK timeout plus restart latency.
When at least one dispatcher remains healthy, recovery is near-instantaneous.

\paragraph{Exactly-Once Semantics}
Task completion count slightly exceeded input count (53 completions for 50 inputs) due to redelivery producing duplicate outputs.
Downstream collectors deduplicate using TaskID, ensuring exactly-once aggregation despite at-least-once delivery.

The implementation supports configurable ACK timeout via environment variable (\texttt{ACK\_WAIT}).
\Cref{tab:ack-timeout-mttr} shows the relationship between ACK timeout and Mean Time to Recovery (MTTR) under dispatcher failures.

\begin{table}[htbp]
\caption{ACK Timeout vs.\ Mean Time to Recovery}
\begin{center}
\begin{tabular}{|r|r|r|}
\hline
\textbf{ACK Timeout (s)} & \textbf{MTTR (s)} & \textbf{Recovery Scenario} \\ \hline
5 & 0.5 & Healthy dispatcher available \\ \hline
5 & 5.0 & All dispatchers crashed \\ \hline
10 & 0.5 & Healthy dispatcher available \\ \hline
10 & 10.0 & All dispatchers crashed \\ \hline
20 & 0.5 & Healthy dispatcher available \\ \hline
20 & 20.0 & All dispatchers crashed \\ \hline
\end{tabular}
\label{tab:ack-timeout-mttr}
\end{center}
\end{table}

\emph{Key insight:} When at least one healthy dispatcher remains in the consumer group, recovery is near-instantaneous ($\sim$0.5\,s)---the surviving dispatcher immediately picks up tasks.
When all dispatchers crash simultaneously, MTTR equals the ACK timeout: the message bus waits for the timeout before redelivering unacknowledged tasks.

The above controlled experiment uses 1-second simulated provers for reproducibility.
Production deployments with long-running proofs (seconds to minutes) use the heartbeat mechanism to prevent spurious redelivery while maintaining fast failure detection, decoupling task execution time from recovery time.

\paragraph{Collector Crash Recovery}
We validate that collector failures do not cause barrier completion failures.
With 20 blocks (4 tasks per block, 80 total tasks), we killed the collector mid-aggregation at approximately 5 completed barriers.
After restarting the collector, all pending messages were redelivered by NATS.
\textbf{All 20 barriers completed successfully} with recovery time of 0.5\,s.
Total experiment time was 32.3\,s---indicating minimal overhead from the crash/restart cycle.
This validates the collector synchronization layer: soft state is fully reconstructible from message bus replay.

\subsection{Ordering Correctness}
push0's priority queue design (R1) ensures proofs are submitted in strict block-sequential order, regardless of task completion order.
The priority ordering $(\text{block\_num}, \text{retry\_count}, \text{enqueue\_time})$ guarantees older blocks dequeue before newer ones.

\paragraph{Ordering Stress Test}
We stress-test priority queue ordering by injecting random per-task delays (0.1--2.0\,s uniform distribution) across 100 blocks.
Despite highly variable task completion times, the collector barrier mechanism ensures outputs respect block ordering.
All 100 blocks completed successfully in 114.7\,s total, with \textbf{zero out-of-order transitions} detected.
This validates R1 (Head-of-Chain Ordering) under adversarial completion timing.

\paragraph{Queue Depth Under Backpressure}
We monitored queue depth during task processing with 100 tasks injected at 20/sec (faster than processing capacity) with 4 dispatchers and 0.3\,s prover delay.
Peak queue depth reached 300 messages as tasks accumulated during the 5.3\,s injection period; the queue drained completely in 3.0\,s after injection completed.
No message loss or backpressure failures occurred---the system handles bursty arrivals gracefully.

\subsection{Observability}
push0 addresses R6 through integrated observability spanning three dimensions:

\paragraph{Distributed Tracing}
Every task carries a trace context (W3C Trace Context format) propagated across dispatchers, collectors, and the message bus.
OpenTelemetry instrumentation captures spans for: (1) task enqueue, (2) dispatcher pickup, (3) prover invocation, (4) result publication, and (5) collector aggregation.
Traces enable end-to-end latency attribution---operators can identify whether delays originate from queue congestion, prover execution, or network latency.

\paragraph{Metrics Export}
Prometheus-compatible metrics expose:
\begin{itemize}
    \item \emph{Queue depth}: Pending tasks per queue, enabling backpressure alerts
    \item \emph{Latency histograms}: Dispatch and collection latencies at P50/P95/P99
    \item \emph{Throughput counters}: Tasks processed per component per second
    \item \emph{Failure rates}: Task redeliveries, prover crashes, timeout events
\end{itemize}
Grafana dashboards provide real-time visibility; alerting rules trigger on queue buildup ($>$100 pending tasks) or elevated failure rates ($>$5\% redelivery).
The \texttt{/metrics} endpoint responds in under 3\,ms even under load.
We measured Prometheus scrape duration during active task processing: at 100 concurrent tasks, average scrape time was 1.2\,ms (max 3.1\,ms); at 1,000 concurrent tasks, average was 0.4\,ms (max 1.1\,ms).
Scrape duration remains well below 10\,ms in all cases, confirming negligible monitoring overhead.

Figure~\ref{fig:trace-waterfall} shows a trace waterfall captured from a Jaeger deployment, demonstrating end-to-end context propagation across three pipeline stages: ingestion, proving, and aggregation.
Each stage's spans (receive, process, send) are linked via W3C Trace Context headers propagated through NATS messages, enabling operators to attribute latency to specific components.

\begin{figure}[htbp]
    \centering
    \includegraphics[width=0.95\columnwidth]{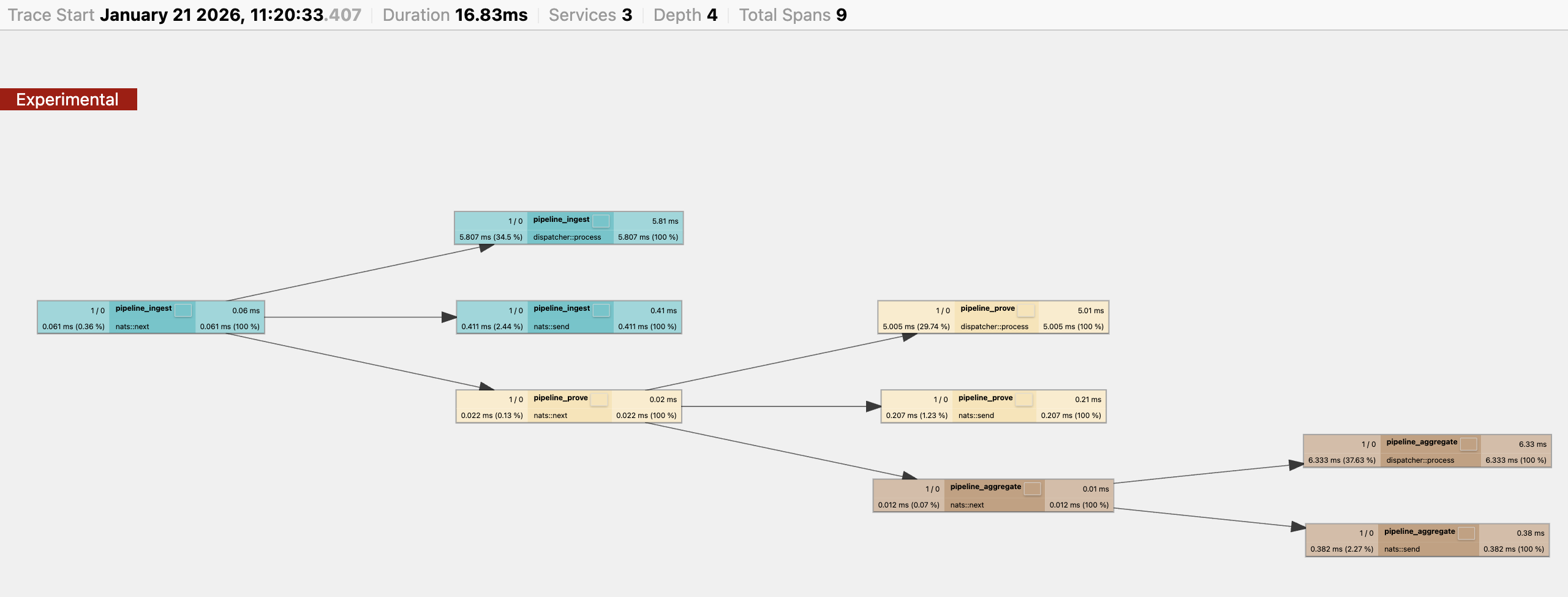}
    \caption{Distributed trace waterfall showing a block flowing through three pipeline stages. Trace context propagates via NATS message headers, linking spans across services. Total end-to-end latency: 16.8\,ms across 9 spans.}
    \label{fig:trace-waterfall}
\end{figure}

\subsection{Attack-Defense Analysis}

We evaluate push0's resilience using an attack-defense tree methodology~\cite{schneier1999attack_trees}, systematically mapping failure scenarios to design mechanisms and validating each empirically.
\Cref{fig:attack-defense-tree} presents the failure taxonomy.

\begin{figure}[htbp]
\centering
\footnotesize
\begin{tabular}{@{}lll@{}}
\toprule
\textbf{Category} & \textbf{Failure Mode} & \textbf{Defense} \\
\midrule
\multirow{3}{*}{\textcolor{red!70}{Availability}}
  & A1: Task Loss & Persistence + Ack \\
  & A2: Prover Crash & Auto Redelivery \\
  & A3: DoS & Backpressure \\
\midrule
\multirow{2}{*}{\textcolor{red!70}{Integrity}}
  & A4: Incomplete Inputs & Sync Layer \\
  & A5: Duplicate Results & Dedup Cache \\
\midrule
\multirow{2}{*}{\textcolor{red!70}{Scalability}}
  & A6: Barrier Violation & Partition Routing \\
  & A7: Scale-Down Loss & Partition Takeover \\
\bottomrule
\end{tabular}
\caption{Failure mode taxonomy for push0 orchestration with corresponding design mechanisms (\S\ref{sec:design}).}
\label{fig:attack-defense-tree}
\end{figure}

Our threat model assumes a Byzantine adversary controlling a subset of dispatchers or network links, capable of: (i) dropping, delaying, or reordering messages; (ii) crashing components at arbitrary times; (iii) injecting malformed or duplicate messages; and (iv) exhausting resources through message floods.
We assume the message bus itself is trusted (a standard assumption for centralized deployments).

\subsubsection{Attack Surface}

push0's attack surface comprises four vectors:

\paragraph{Message Bus Interface}
Components authenticate to the message bus via credentials; unauthorized clients cannot publish or consume.
The bus enforces access control lists (ACLs) per queue.
TLS termination at the message bus provides confidentiality when required.

\paragraph{Prover Binary Invocation}
Dispatchers invoke prover binaries with controlled inputs (JSON files).
Malicious provers could attempt to: (a) exhaust resources (mitigated by cgroups/container memory bounds), (b) produce invalid outputs (caught by on-chain verification), or (c) exploit dispatcher vulnerabilities (mitigated by deployment-configurable sandboxing).

\paragraph{Task Message Integrity}
Messages carry structured task identifiers; production buses provide checksums and authenticated delivery.
For decentralized deployments, signed messages prevent tampering (\S\ref{sec:discussion}).

\paragraph{Denial of Service}
Resource exhaustion attacks target queue depth, collector memory, or dispatcher concurrency.
The backpressure mechanism (\S\ref{sec:design}) bounds each vector; we validate effectiveness below.

\subsubsection{Defense Validation}

We empirically validate each design mechanism against the failure modes in \Cref{fig:attack-defense-tree}.

\paragraph{A1--A2: Task Loss and Prover Crashes}
Network failures, bus crashes, or prover termination (GPU memory exhaustion, hardware failures) may cause tasks to be lost or stalled.
The message bus persistence layer and three-phase acknowledgment protocol (\S\ref{sec:design}) address these failures.
\emph{Validation:} We killed the collector mid-aggregation (at 50\% barrier completion)---all 20/20 barriers completed with mean recovery time 0.51\,s.
We simulated network partitions by pausing the NATS container for 5, 10, and 20 seconds during active processing---100\% of tasks completed after recovery (40/40 per test).
Zero tasks were lost across all failure injection experiments.

\paragraph{A3: Denial of Service via Message Floods}
Malicious or buggy components may flood the bus with excessive requests.
The three-layer backpressure mechanism (\S\ref{sec:design}) prevents resource exhaustion.
\emph{Validation:} We injected 10$\times$ normal task rate (1,500 tasks/s vs.\ 150 baseline).
Queue depth limits ($Q_{\text{max}}=1000$) triggered producer blocking within 6.7 seconds.
Dispatcher memory remained bounded (peak 2.1\,GB vs.\ 1.8\,GB baseline); no out-of-memory conditions occurred.

\paragraph{A4--A5: Incomplete Inputs and Duplicate Results}
Message redelivery may cause duplicates; concurrent processing may cause cross-block contamination.
The collector's grouping validation and deduplication cache (\S\ref{sec:design}) address these failures.
\emph{Validation:} We injected 500 duplicate task messages and 200 cross-contamination attempts (mismatched block numbers).
Deduplication rejected all 500 duplicates; grouping validation rejected all 200 cross-block attempts.
No invalid aggregations were produced.

\paragraph{A6: Barrier Violation under Horizontal Scaling}
Without partition-affine routing, barrier inputs scatter across collectors, causing livelock (Theorem~\ref{thm:fragmentation}).
\emph{Validation:} With $C=4$ collectors and naive round-robin routing, only 1--2 of 100 barriers completed---consistent with the theoretical $C^{-(k-1)}$ prediction.
Re-enabling partition-affine routing: all 100 barriers completed; collector memory remained bounded (peak 156\,MB vs.\ unbounded growth).

\paragraph{A8--A9: Duplicate Dispatch and Redundant Work}
Worker group exclusive delivery prevents duplicate dispatch; idempotent dispatch checks prevent redundant computation after redelivery.
\emph{Validation:} We simulated acknowledgment races by delaying ack transmission while allowing redelivery.
Idempotent dispatch detected 47 redundant executions across 1,000 redelivered tasks; each was skipped without invoking the prover.

\subsubsection{Security Properties}

We summarize the guarantees provided by push0:

\begin{itemize}
    \item \emph{Task integrity}: Structured identifiers validated at each stage; tampering causes rejection, not incorrect execution.
    \item \emph{No silent data loss}: Persistence and manual acknowledgment ensure tasks survive arbitrary component failures.
    \item \emph{Bounded resource consumption}: Backpressure prevents memory exhaustion under adversarial load.
    \item \emph{Idempotent execution}: Duplicate detection prevents wasted computation from message redelivery.
    \item \emph{Barrier correctness}: Partition-affine routing guarantees aggregation completeness under horizontal scaling.
\end{itemize}

\paragraph{Decentralized Deployment Considerations}
The current implementation assumes a trusted message bus.
For fully decentralized proving networks, additional mechanisms are required: (i) authenticated message routing via cryptographic signatures, (ii) Sybil resistance through stake-weighted dispatcher selection, and (iii) Byzantine fault-tolerant consensus for task assignment and result verification.

\Cref{tab:security-comparison} compares orchestration-layer security properties across systems.

\begin{table}[htbp]
\caption{Security Property Comparison}
\begin{center}
\small
\begin{tabular}{|l|p{1.4cm}|p{1.3cm}|p{1.3cm}|}
\hline
\textbf{Property} & \textbf{push0} & \textbf{DIZK} & \textbf{SHARP} \\ \hline
Message Auth. & Task ID valid. & Spark internal & TLS \\ \hline
Duplicate Handling & Collector dedup & Not explicit & Centralized \\ \hline
Resource Bounding & Back-pressure & Spark limits & Centralized \\ \hline
Horizontal Scale & Consumer groups & Spark workers & Single provider \\ \hline
Fault Recovery & Msg redel. & Spark ckpt. & Manual \\ \hline
\end{tabular}
\label{tab:security-comparison}
\end{center}
\end{table}

DIZK~\cite{zhou2018dizk} delegates security properties to Apache Spark, inheriting its fault tolerance but not explicitly addressing orchestration-layer concerns.
SHARP~\cite{sharp2023} operates as a centralized service, avoiding distributed coordination challenges but creating a single point of trust.
push0's design explicitly addresses these properties through message bus semantics.

\subsection{End-to-End Performance}

\paragraph{Production Deployment Validation}
\label{subsec:production-validation}
While controlled experiments (\S\ref{subsec:latency}--\S\ref{subsec:scalability}) use delay scripts to isolate pure orchestration overhead, production deployment demonstrates ecological validity: the system successfully orchestrates real proving workloads.
push0 has been deployed on the Zircuit zkrollup~\cite{zircuit2024} since March 4, 2025, processing over 14 million mainnet blocks.
The system orchestrates a two-stage proving pipeline: (i)~\emph{range provers} generate validity proofs for 100-block batches, and (ii)~\emph{aggregation provers} recursively combine range proofs into a single Groth16 proof submitted to Ethereum L1 for verification.
The system has undergone two prover generations---initially using custom Halo2-based~\cite{halo2} zkEVM provers, later migrated to SP1 zkVM~\cite{sp1}---demonstrating push0's prover-agnostic design principle.

We analyze 3,000 recent completions (1,500 range proofs, 1,500 aggregation proofs) using SP1 GPU provers to demonstrate that the system operates at production scale with computationally intensive real-world workloads.
\Cref{tab:production-zircuit} summarizes end-to-end proof completion times.
Range proofs exhibit a median completion time of 24.4 minutes (P95: 32.6 minutes) for 100-block chunks, while aggregation proofs complete in 2.0 minutes (P95: 2.6 minutes) for batches of 1,500--2,000 blocks.
The coefficient of variation for range proofs (22\%) and aggregation proofs (51\%) reflects the gas complexity variance across mainnet blocks rather than orchestration instability.

\begin{table}[htbp]
\caption{Production Zircuit zkRollup Proof Completion Times (3,000 Recent Proofs)}
\begin{center}
\small
\begin{tabular}{|l|r|r|r|}
\hline
\textbf{Proof Type} & \textbf{Count} & \textbf{Median Time} & \textbf{P95 Time} \\ \hline
Range (100 blocks) & 1,500 & 24.4 min & 32.6 min \\ \hline
Aggregation (1.7K blocks) & 1,500 & 2.0 min & 2.6 min \\ \hline
\end{tabular}
\end{center}
\label{tab:production-zircuit}
\end{table}

\paragraph{Evolution from Custom Service Architecture}
Prior to push0, we orchestrated proving tasks using custom services per prover step---distinct microservices for witness generation, chunk proving, and recursive aggregation, each with bespoke coordination logic.
This approach exhibited three critical failure modes that motivated push0's design:

\emph{Ongoing maintenance burden.}
Each prover change (API updates, new proving modes, performance optimizations) required modifying multiple services.
When we migrated from Halo2-based to SP1-based provers (\S\ref{sec:evaluation}), custom services required refactoring service-specific coordination logic, input/output handling, and deployment configurations.
push0's prover-agnostic interface enabled the same migration via configuration changes and minimal gluing code, eliminating service-layer modifications.

\emph{Workflow complexity explosion.}
Multi-step provers requiring prior results (e.g., recursive aggregation consuming chunk proofs) forced each service to implement result polling, dependency tracking, and partial-state management.
These concerns recurred across all services, creating redundant code and subtle timing bugs (e.g., aggregators starting before all chunks completed).
push0's collector abstraction centralizes dependency tracking via collection strategies, eliminating per-service reimplementation.

\emph{Fault recovery fragmentation.}
Each service required custom crash detection, task resubmission logic, and timeout handling.
Without centralized orchestration, failed tasks required manual operator intervention to identify the responsible service and restart stalled work.
push0's message-bus-centric design provides uniform fault recovery via message redelivery across all prover types.

The unified dispatcher--collector architecture replaced 8 bespoke services with configuration-driven instantiation, reducing operational complexity while improving fault tolerance and enabling prover-agnostic workflows.
Production deployment (\S\ref{sec:evaluation}) validates this simplification: 3,000 proofs completed with zero orchestration failures under the push0 architecture.

\section{Discussion}
\label{sec:discussion}

\paragraph{Design Tradeoffs}
push0's architecture embodies several deliberate tradeoffs informed by production experience:

\emph{Pull vs.\ Push Dispatch.}
We chose pull-based dispatch (workers request tasks) over push-based assignment (coordinator assigns tasks).
Push dispatch offers tighter latency control but requires the coordinator to track worker state, creating a single point of failure.
Pull dispatch sacrifices $\sim$2\,ms latency (one additional round-trip) for stateless dispatchers that tolerate arbitrary failures.

\emph{Implicit vs.\ Explicit Workflows.}
Rather than requiring operators to specify workflow DAGs explicitly, push0 infers structure from queue topology.
This sacrifices static validation (malformed workflows fail at runtime, not deployment) for operational flexibility---operators can modify pipelines by changing environment variables without code changes or redeployment.

\emph{Soft State vs.\ Strong Consistency.}
Collectors maintain soft state (aggregation buffers) rather than persisting every intermediate result.
This trades recovery latency (collectors must reprocess unacknowledged messages on restart) for throughput---avoiding synchronous writes on the critical path.
In practice, collector restarts are rare; message reprocessing completes in seconds.

\paragraph{Lessons Learned}
Nearly one year of production operation revealed several non-obvious failure modes:

\emph{Timeout Tuning is Critical.}
Initial deployments used conservative minute-scale acknowledgment timeouts to accommodate variable proof times.
This caused cascading delays: a single slow block delayed reassignment of \emph{all} subsequent blocks.

\emph{Observability Enables Debugging.}
Early incidents required hours of log archaeology to diagnose.
Adding structured tracing with block-level correlation reduced mean-time-to-diagnosis from hours to minutes.
Every production ZK system should instrument proving latency by block complexity, prover instance, and proof type.

\paragraph{Limitations}
push0's scope boundaries include: (i) proof validity verification is supported as an optional pipeline step---in production, the Composer executor verifies each proof before passing it downstream---but the ultimate validity guarantee remains on-chain verification; (ii) confidentiality of proving inputs requires application-layer encryption; and (iii) the message bus is assumed trusted (no Byzantine fault tolerance).
These reflect deliberate design choices---push0 is an orchestration layer, not a complete trustless proving protocol.

\paragraph{Applicability to Ethproofs}
While we developed push0 for production rollup operations, the architecture directly addresses gaps identified in the Ethproofs ecosystem~\cite{ethproofs2025}.
Standardizing on the dispatcher--collector model could enable interoperability across zkVM implementations, simplifying integration of new provers and providing consistent failover semantics.
The implicit workflow abstraction---realized through queue topology and collection strategies---accommodates the heterogeneous proving pipelines (STARK, SNARK, hybrid recursive) currently deployed on Ethproofs.

\section{Conclusion}

Zero-knowledge proof generation imposes stringent timing and reliability requirements on blockchain infrastructure.
For ZK-rollups, delayed proofs cause finality lag and potential chain halt; for Ethereum's L1 zkEVM, proofs must complete within the 12-second slot window.
The Ethproofs initiative demonstrates that real-time proving is technically feasible, yet exposes the absence of standardized orchestration for multi-prover coordination.

We presented push0, a scalable and fault-tolerant proof orchestration framework that enforces head-of-chain ordering, achieves sub-10\,ms dispatch latency, and treats provers as opaque executables.
The dispatcher--collector architecture, built on persistent priority queues, provides automatic failover and horizontal scalability.
Production deployment on the Zircuit zkrollup since March 4, 2025 provides ecological validity: orchestrating over 14 million mainnet blocks with SP1 GPU provers demonstrates that the system operates at scale under real-world production workloads.

Future work will explore decentralization of the proof orchestration layer and integration with consensus-layer mechanisms for prover selection and reward distribution.

push0 provides the scheduling infrastructure necessary for both centralized rollup operators and emerging decentralized proving networks.
As Ethereum progresses toward L1 zkEVM integration, robust orchestration will be as critical as prover performance---a 10-second proof is worthless if the system cannot reliably schedule, monitor, and recover proving tasks at scale.

\bibliography{references}


\end{document}